\documentclass[twoside,a4paper,final]{article}
\usepackage[left=3.00cm, right=3.00cm, top=2.00cm, bottom=2.00cm]{geometry}
\usepackage{amssymb,amsthm,amsmath}
\newtheorem{definition}{Definition}[section]
\newtheorem{theorem}{Theorem}[section]
\newtheorem*{thmA}{Theorem A}
\newtheorem*{lemmaA}{Subgroup Decomposition Lemma}

\newtheorem{lemma}[theorem]{Lemma}
   
\newtheorem{corollary}[theorem]{Corollary}
\newtheorem{remark}{Remark}[section]
\usepackage{cite}
\usepackage[scaled]{helvet}
\theoremstyle{remark}
\title{Niederreiter cryptosystems using quasi-cyclic codes that resist quantum Fourier sampling}
\author{Upendra Kapshikar\and Ayan Mahalanobis}
\newcommand{\Addresses}{{
  \bigskip
  \footnotesize
  
U.~Kapshikar, \textsc{Center for Quantum Technologies and National Univeristy of Singapore}\par\nopagebreak
  \textit{E-mail address}, U.~Kapshikar: \texttt{uskapshikar@gmail.com}
\medskip

A.~Mahalanobis (Corresponding author), \textsc{IISER Pune, Pune, INDIA}\par\nopagebreak
  \textit{E-mail address}, A.~Mahalanobis: \texttt{ayan.mahalanobis@gmail.com}
}}
\providecommand{\keywords}[1]{\textbf{\textit{Keywords--}} #1}
\begin{document}
\date{}
\maketitle
\begin{abstract}
  McEliece and Niederreiter cryptosystems are robust and versatile cryptosystems. These cryptosystems work with many linear error-correcting codes. They are popular these days because they can be quantum-secure. In this paper, we study the Niederreiter cryptosystem using non-binary quasi-cyclic codes. We prove, if these quasi-cyclic codes satisfy certain conditions, the corresponding Niederreiter cryptosystem is resistant to the hidden subgroup problem using weak quantum Fourier sampling. Though our work uses the weak Fourier sampling, we argue that its conclusions should remain valid for the strong Fourier sampling as well. 
\end{abstract}
\section{Introduction}
McEliece and Niederreiter cryptosystems are important in \textbf{post-quantum cryptography}, they can resist \textbf{quantum Fourier sampling} (QFS) attacks, which is the main ingredient in the Shor's algorithm~\cite{Shor}. The idea of (in)effectiveness of QFS was introduced by Kempe and Shalev~\cite{Kempe}, where they characterized subgroups in permutation groups that can be distinguished from the identity subgroup via a quantum Fourier sampling. 

Using their idea, Dinh et al.~\cite{Dinh} showed, if the generator matrix of a McEliece cryptosystem is \textbf{well-scrambled} and \textbf{well-permuted}\footnote{For more details on well-scrambled and well-permuted, see Dinh et al.~\cite{Dinh}, after Corollary 1 in Section 4.2.}, it is resistant to QFS. Kapshikar and Mahalanobis~\cite{Kapshikar} used the same theorem to show that when certain kind of quasi-cyclic codes are used as the parity-check matrix of a Niederreiter cryptosystem, it is resistant to QFS.

In this paper, we extend the result of Kapshikar and Mahalanobis~\cite{Kapshikar} to propose an enhanced set of quasi-cyclic codes with a more natural set of conditions, on which the Niederreiter cryptosystem is resistant to QFS. In particular, condition (IV) in the earlier work imposed a group theoretic condition on the parity-check matrix $\mathcal{H}$. We remove that completely. Furthermore, our earlier work was based on rate $\frac{m-1}{m}$ quasi-cyclic code. This work is on arbitrary quasi-cyclic code.
 
The \textbf{main motivation for this research} is to explore quantum-security for Niederreiter cryptosystems using \textbf{non-binary quasi-cyclic codes}. There is a lot of interest in McEliece and Niederreiter cryptosystems using quasi-cyclic codes~\cite{baldi_07, zurich, Baldi, baldi_book,rosen,otmani2010cryptanalysis,gaborit1,gaborit}.  In this paper, we propose Niederreiter cryptosystems on non-binary quasi-cyclic codes that resists quantum Fourier sampling. This way we can use some nice features of quasi-cyclic codes, like its fast decoding algorithm and smaller key sizes along with the idea of quantum security via the hidden subgroup problem that was in the original work of McEliece and Niederreiter. This work complements the work of Baldi et al.~\cite{baldi_NB}.

Our main theorem is the following:
\begin{thmA}\label{main_thm}[Main Theorem] For quasi-cyclic codes, which satisfy conditions (i) - (iv) in Section~\ref{cond_qcodes}, the corresponding Niederreiter cryptosystems resist quantum Fourier sampling attack.
\end{thmA}

Our proof can be divided up into 3 parts. Let $\mathcal{H}$ be the parity-check matrix of size $k\times n$ to be used in the Niederreiter cryptosystem. We write $\mathcal{H}=[\mathrm{I}|C]$, where $C$ is a block
matrix with each block being a circulant matrix of size $d$ over a proper extension of $\mathbb{F}_{2}$. 
\begin{description}
\item[A)] First we show that for any codes, the hidden subgroup problem over $\left( \mathrm{GL}_{k}(\mathbb{F}_2) \times \mathrm{S}_n \right)^{2} \rtimes \mathbb{F}_{2} $ can be broken down into
  $\mathrm{GL}_{k}(\mathbb{F}_{2})$ and $\mathrm{S}_n$ with an additional overhead in terms of the size of the hidden subgroup. Note that this decomposition into individual components is true for any error
  correcting code. It can be useful for other codes as well.   
\item[B)] In the second part we use the structure of quasi-cyclic codes to find some important bounds. Here we also show, for our codes, the automorphism group is contained in $\mathrm{S}_k \times \mathrm{S}_n$. This makes it easier to frame the automorphism group as a stabilizer where the group action is reordering rows and columns.  
\item[C)] We finish by combining the above arguments with some important results from the theory of permutation groups to show that the hidden subgroup can not be distinguished from the identity subgroup.
Hence the Niederreiter cryptosystem thus constructed will be resistant to quantum Fourier sampling. 
\end{description}
\paragraph{Strong vs.~weak Fourier sampling:}
The standard version of QFS has two main variations -- the weak QFS and the strong QFS. Both these versions deal with complex irreducible linear representation of a group.
In the weak QFS only the irreducible representation $\rho$ is measured whereas in the strong form along with the representation, entry inside the corresponding representation, denoted by $\left(\rho,i,j\right)$, is also measured. It is often assumed that the strong QFS could provide significantly more information. Note that, in the case of abelian groups all complex irreducible representations are one dimensional so the weak and the strong form are the same. In the case of some non-abelian groups, such as the dihedral group $\mathrm{D}_{2n}$,  measuring the strong form $\left(\rho, i, j\right)$ is necessary to obtain enough information for distinguishing particular subgroups like subgroups generated by involutions. Interestingly, for many groups such as the symmetric group $\mathrm{S}_n$ the picture is not so clear. First note that there is an inherent selection that one must do in the strong form, namely the basis for the representation. Fourier transforms can only be defined uniquely only upto a basis. So both the weak and the strong methods have different QFS in different basis for each irreducible representation. This question of a gap between the strength of the weak form and the strong form of QFS was addressed by Grigni et al.~\cite{Vazirani}. They proved the following:
\begin{description}
\item[D)] In the strong form wherein the algorithm is supposed to measure $\left(\rho, i, j\right)$, there is no point in measuring the row index $i$. In other words, measuring $\left(\rho,j\right)$ would provide the same amount of qualitative information as given by $\left(\rho, i, j\right)$~\cite[Section 3.1]{Vazirani}.
\item[E)] For groups that have small number of conjugacy classes as compared to the size of the group, as long as the hidden subgroup is not too big, strong QFS in a random basis fails to give any advantage over its weak counterpart~\cite[Theorem 3.2]{Vazirani}.  
\end{description}   

Kempe and Shalev~\cite{Kempe} took this idea forward and characterized subgroups of $\mathrm{S}_n$ that can be distinguished from the identity subgroup in the weak form. It follows from E) above that the characterization remains true in the strong QFS when done with a random basis. Note that as E) is based on a random choice of basis, in principle there could be a \emph{clever} choice of basis that somehow works magically. However, till this date there is no known example of this where the choice of basis has led to a significant improvement. So it is believed that the strong form is unlikely to provide stronger results. Now in the scrambler-permutation problem, the group in question is the semidirect product $\left(\mathrm{GL}_{k}(\mathbb{F}_{2}) \times \mathrm{S}_n\right)^2 \rtimes \mathbb{Z}_{2}$. Similar to $\mathrm{S}_n$, this group is also highly non-commutative so a random basis QFS will fail to provide any significant improvement. With this understanding, we prove Theorem A for weak QFS which extends to a random basis strong QFS by the the above argument. 
\subsection{Structure of the paper}
In Section~\ref{st_method}, we introduce briefly the \emph{standard method} for quantum Fourier sampling and associated hidden subgroup problem along with a result by Kempe and Shalev~\cite{Kempe,KS2}. In the following
section, we briefly talk about code based cryptosystems and in particular the Niederreiter cryptosystem. We also mention the main quantum attack known as \emph{the scrambler-permutation attack} and its connection to
the hidden subgroup problem which is a main objective of this paper. 
Scrambler-permutation attacks are important. In general $\mathcal{H}$, a part of the private key, is not known. However, if there is a proof that with the knowledge of the $\mathcal{H}$ the quantum Fourier sampling is indistinguishable from the uniform distribution then that remains true without the knowledge of $\mathcal{H}$. This provides great confidence in the quantum security of a Niederreiter cryptosystem. Then we move to the main contribution of this paper -- proof of Theorem A. 
\subsection{A context for Theorem A} McEliece and Niederreiter developed public-key cryptosystems using linear codes. The biggest issue with these cryptosystems is the key-size. One main reason people looked at quasi-cyclic codes for these cryptosystems is reduction of key sizes. The first mention of quasi-cyclic codes in McEliece cryptosystem was Monico et al.~\cite{rosen}. Then Gaborit~\cite{gaborit1} used classical representation of quasi-cyclic codes and primitive BCH codes as the particular instance to suggest reducing the key-size of McEliece and Niederreiter cryptosystems. Though that particular instance of a McEleice cryptosystem was quickly broken down by Otmani et al.~\cite{otmani2010cryptanalysis}. The idea is still interesting.

Though the main motivation of our work is resistance against quantum Fourier sampling. The right context to look at this work is non-binary quasi-cyclic codes in Niederreiter cryptosystems.  There are two ways of representing quasi-cyclic codes. One is the classical representation~\cite[Section 3.1]{baldi_book} and the other is the alternative representation~\cite[Section 3.3]{baldi_book}.

LEDAcrypt~\cite{nist1} and QC-MDPC~\cite{nist2} cryptosystems were built using the alternative representation of quasi-cyclic codes. The main theorem (Theorem A) in this paper uses the alternative representation. Assuming that enough circulant matrices are invertible, one can switch between the classical and the alternative representation using row-column operations. So, from the point of view of the scrambler-permutation problem, the problem is invariant under representation. We present our Niederreiter cryptosystem using the classical representation. Since LEDAcrypt and QC-MDPC uses a different formulation for their Niederreiter cryptosystem, our theorem does not apply to these cryptosystems.

Non-binary LDPC codes can have efficient decoding algorithms~\cite[Section V]{baldi_NB}~\cite{declercq}. They are considered better codes as they can correct more errors. If we assume no structure in the scrambler and permutation matrices then the public key will have no structure. This will make the cryptosystem resistant to ISD attacks. An ISD attack was recently found on LEDAcrypt~\cite{LEDAattack}. The fact that the non-binary LDPC code can be a better code might reduce the key size marginally and offer a fair competition to the existing McEliece and Niederreiter cryptosystems. This point was echoed by Baldi et al.~\cite{baldi_NB}.

Now going back to the idea of Gaborit, one can take the scrambler and the permutation matrix in the Niederreiter cryptosystems as a block diagonal matrix and a block permutation matrix with identical blocks respectively. If we do that and take the private-key $\mathcal{H}$ as the parity-check matrix of a non-binary quasi-cyclic code then the key size is reduced significantly as argued by Gaborit. Whether it remains secure is an open question. 
\section{Preliminaries}In this paper we use some fundamental concepts from finite group theory. In particular, we use the concept of a group acting on itself by conjugation. We use the idea of an orbit which is the conjugacy class and the stabilizer which is the centralizer. We further use the orbit-stabilizer theorem.
For details on these concepts the reader can consult any standard book on group theory, like Alperin and Bell~\cite{bell} or Dixon and Mortimer~\cite{Dixon}.
\subsection{Hidden Subgroup Problem}
One of the unifying theme in quantum computing is the \emph{hidden subgroup problem}. Most of the practical algorithms that offer exponential speedup in quantum computer science can be modeled in this form.
Popular examples are, factoring integers by Shor's algorithm, the discrete logarithm problem and others~\cite{hallgreen}.  

\begin{definition}[Hidden Subgroup Problem] Let $\mathrm{G}$ be a group and $\mathrm{H}$ be a (unknown) subgroup of $\mathrm{G}$. We are given a function $f$ from $\mathrm{G}$ such that $f(g_{1})=f(g_{2})$
  whenever $g_{1}\mathrm{H}=g_{2}\mathrm{H}$. The function in this case is said to be separating cosets of the subgroup $\mathrm{H}$. The hidden subgroup problem is to find a set of generators of $\mathrm{H}$.
\end{definition}

The hidden subgroup problem is easy to solve when the group $\mathrm{G}$ is abelian but for non-commutative groups it is far from realized. Efforts to solve the hidden subgroup problem can be broadly
characterized into two categories. One of which is based on a generalization of quantum Fourier sampling from abelian to non-abelian groups~\cite{Roetteler1998,hallgren2}. The second direction is on some
particular non-abelian black-box groups where instead of doing quantum Fourier transform over the group, it is done in some abelian group~\cite{Santha, John}. Apart from this, some strong structural results
are available in the non-commutative case, see Vazirani~\cite{Vazirani}. In this paper, we follow the first approach and our \emph{indistinguishability depends on the quantum Fourier transforms on
  non-abelian groups}. The function for the Fourier transform is given by an irreducible representation over the field of complex numbers. 

\subsection{Quantum Fourier Sampling}\label{st_method}
Algorithms based on QFS were developed based on the \emph{standard method} by Simons and Shor~\cite{Shor}. We roughly sketch the process. The quantum Fourier sampling is based on a unitary transformation
defined as follows:
\begin{definition}[QFT] A quantum Fourier transform takes an element of the group algebra $\mathbb{C}[\mathrm{G}] $ to the representation basis or the Fourier basis for a group $\mathrm{G}$.

\[ QFT(g)= \dfrac{1}{\sqrt{\vert \mathrm{G} \vert}}\sum_{\rho,i,j} \sqrt{d_\rho}\rho(g)_{i,j} (\rho,i,j)\]
where $\rho$ is an irreducible representation of $\mathrm{G}$ and $d_\rho$ is its dimension.
\end{definition} 
Here we briefly recall the standard method for QFS. For detailed version we refer~\cite{Kempe,Vazirani}. Initialise the state in the superposition of all states in the first register and $\vert 0 \rangle$ in the second register.  Compute $f$ via the oracle $\mathcal{O}_f$ defined as $\mathcal{O}_f(\vert x, 0\rangle):= \vert x, f(x)\rangle$. This is followed by measurement on the second register which puts the state of the first register in a random left coset of a subgroup $\mathrm{H}$ i.e.,
$\vert g\mathrm{H} \rangle = \sum_{h\in \mathrm{H}} \vert gh \rangle $ for a random $g$. Finally QFT along with the measurement in the Fourier basis, gives the probability distribution as 
\begin{equation}\mathrm{P}_{g\mathrm{H}}\left( \rho,i,j \right) = \dfrac{\vert d_\rho\vert} { \vert \mathrm{G} \vert\ \vert \mathrm{H} \vert }  {\left\vert \sum_{h \in \mathrm{H}} \rho(gh)_{i,j} \right\vert}^{2}.
\end{equation}
As we have chosen $g$ randomly and uniformly, $\mathrm{P}_{\mathrm{H}}= \frac{1}{\vert \mathrm{G}\vert} \sum_{{g}} \mathrm{P}_{g\mathrm{H}}$. In the weak version of QFS the measurement is done only on labels for irreducible representation. \begin{equation} \mathrm{P}_{\mathrm{H}}\left(\rho\right)= \dfrac{d_{\rho}}{\vert \mathrm{G} \vert} \sum_h {\chi_\rho(h)}\end{equation} where $\chi_\rho$ is the character. 

To solve the hidden subgroup problem, it suffices to find $\mathrm{H} $ in
$\text{poly}(\log(\vert \mathrm{G}\vert))$ time. To distinguish $\mathrm{H}$ from the identity subgroup $\langle e \rangle$, it is necessary that $L_{1}$-distance between
$\mathrm{P}_{\mathrm{H}}$ and $\mathrm{P}_{\langle e\rangle}$ is greater than some inverse polynomial in  $\log(\vert \mathrm{G}\vert)$. Thus $\mathrm{H}$ is distinguishable from $\langle e \rangle$ if
there exists a constant $c$ such that  $ \mathcal{D}_{\mathrm{H}}:=  \Vert \mathrm{P}_{\mathrm{H}} - \mathrm{P}_{\langle e \rangle}\Vert_{1} \geq \left(\log\vert \mathrm{G}\vert\right)^{-c}$. Otherwise, we say
that $\mathrm{H}$ is not distinguishable from $\langle e \rangle$. So,  if for all constants $c$,  $\mathrm{H}$ and $\langle e\rangle$ have $L_1$-distance smaller than ${(\log(\vert \mathrm{G}\vert))}^{-c}$,
we say that $\mathrm{H}$ is not distinguishable from the identity subgroup. For more on this we refer to Kempe and Shalev~\cite{Kempe}. It is well known that
\begin{equation} {\label{KS1.1}}\mathcal{D}_{\mathrm{H}} \leq \sum_{i}\vert C_{i} \cap \mathrm{H} \vert \vert C_{i}\vert^{-\frac{1} {2}} = \sum_{h\in \mathrm{H}, h \neq e} \vert {h^{\mathrm{G}}\vert}^{-{\frac{1}{2}}}
\end{equation}
 where $C_{i}$ is a non-identity conjugacy class of $\mathrm{G}$ and $h^{\mathrm{G}}$ denotes conjugacy class of $h$ in $\mathrm{G}$.
 
So, by showing $\mathcal{D}_{\mathrm{H}}$ is less than every inverse polynomial in $\log(\vert \mathrm{G} \vert)$ one can show that, QFS can not successfully reveal the hidden subgroup $\mathrm{H}$. This is how we proceed in this paper, building on the work of Kempe and Shalev~\cite{Kempe}.
\section{Code based cryptosystems}
There is a natural association between coding theory and cryptography because coding theory is a source of many computationally hard problems. More importantly, it is one of the promising areas in post-quantum cryptography as the underlying structure is non-commutative. As mentioned earlier, Shor-like algorithms, that are based on QFS are very effective over abelian groups. So, cryptosystems based on non-commutative groups are thought to be potential candidates for post-quantum cryptography. 

One of the earliest cryptosystems based on error correcting codes was by McEliece~\cite{ME}. A similar cryptosystem was proposed by Niederreiter~\cite{Niederreiter}. Later, a signature scheme based upon Niederreiter systems was also presented. The system we consider is a Niederreiter cryptosystem, based on quasi-cyclic codes.  
\subsection{Niederreiter cryptosystems}
\noindent Let $\mathcal{H}$ be a $k \times n$ parity-check matrix for a $[n,n-k]$ linear code $\mathcal{C}$ over $\mathbb{F}_{2^\eta}$ for $\eta > 1$ with a fast decoding algorithm. Let $e$ be the number of errors that $\mathcal{C}$ can correct.  

\noindent\textbf{Private Key}: ($A_0$,$\mathcal{H}$,$B_0$) where $A_0 \in \rm{\rm{GL}}_{k}(\mathbb{F}_{2})$ and  $B_0$ is a permutation matrix of size $n$.

\noindent\textbf{Public Key}: $\mathcal{H}^{\prime}=A_0\mathcal{H}B_0$.

\noindent\textbf{Encryption}: 
\begin{description}
\item Let $\mathcal{X}$ be a $n$-bit plaintext with weight at most $e$. The corresponding ciphertext $\mathcal{Y}$ of $k$-bits is obtained by calculating $\mathcal{Y}=\mathcal{H}^\prime\mathcal{X}^{\rm{T}}$.

\end{description}
\textbf{Decryption}:
\begin{description}	
\item Compute $y=A_0^{-1}\mathcal{Y}$. Thus $y = \mathcal{H}B_0\mathcal{X}^{\rm{T}}$.
\
\item Now use a fast decoding algorithm for $\mathcal{C}$ with matrix $\mathcal{H}$ and vector $y$ to get $B_0 \mathcal{X}^{\rm{T}}$ and recover $\mathcal{X}$.
\end{description}
We abuse the notation lightly and use $e$ for the identity element in a group as well.
\paragraph{Scrambler-Permutation Attack:}
Scrambler-permutation attacks are defined as, given $\mathcal{H}$ and $\mathcal{H}^\prime$, find $A_0$ and $B_0$. Note that any $A \in\mathrm{GL}_{k}(\mathbb{F}_2), B \in S_n$ that satisfy $\mathcal{H}^{\prime}=A\mathcal{H}B$ breaks the system.
Quantum computers, in principle, can exploit this attack. This follows from the fact that scrambler-permutation attacks can be reduced to a hidden subgroup problem. As we saw in previous sections, hidden
subgroup problem is important because quantum computers have an advantage over classical computers for this class of problems over abelian groups. 

In a scrambler-permutation attack we assume that $\mathcal{H}$ and $\mathcal{H}^\prime$ are known. However, this is not true in practice, only the public key $\mathcal{H}^\prime$ is known. But, if we can show that with the knowledge of $\mathcal{H}$, the cryptosystem is resistant to the quantum Fourier sampling then it is true without the knowledge of $\mathcal{H}$. 
To illustrate the reduction to hidden subgroup problem, we first define a problem that is very close to the hidden subgroup problem.
\begin{definition}[Hidden shift problem]  Let $f_{0},f_{1}$ be two functions from a group $\mathrm{G}$ to some set $\mathrm{X}$ such that: there is a $g_{0}$ such that for all $g$, $f_{0}(g)=f_{1}(g_0g)$. The
  hidden shift problem is to find one such $g_0$. 
\end{definition}
One can frame the scrambler-permutation attack as a hidden shift problem over $\mathrm{G}= \mathrm{GL}_{k}\left(\mathbb{F}_{2}\right) \times \mathrm{S}_n$ where $f_{0}(A,B)=A^{-1}\mathcal{H}B$ and $f_{1}(A,B)=A^{-1}\mathcal{H}^{\prime}B$. Moreover, it is known that for any non-commutative group $\mathrm{G}$, a hidden shift problem can be reduced to a hidden subgroup problem on $\mathrm{G}^{2} \rtimes \mathbb{F}_{2}$ where the
action of $1$ on $(x,y)$ is $(y,x)$. In this paper we are interested in the particular case of $\mathrm{G}=\mathrm{GL}_{k}(\mathbb{F}_{2})\times \mathrm{S}_n$. For that we refer to Dinh et al.~\cite[Proposition 3]{Dinh}, we use their notations for important subgroups for easy reference. The hidden subgroup is
\begin{equation}\label{hidden}
K=\left((H_0,s^{-1}H_0s),0\right)\cup\left((H_0s,s^{-1}H_0),1\right)
\end{equation}
where $H_0=\{(A,P)\in\mathrm{GL}_k(\mathbb{F}_2)\times \mathrm{S}_n\,:\,A^{-1}\mathcal{H}P=\mathcal{H}\}$. Here $s=(A_0^{-1},B_0)$ is the shift.
Note that from now onward we will use $\mathrm{H}$ for $H_0$.

Thus to break a Niederreiter cryptosystem using QFS, one needs to solve a hidden subgroup problem over $\mathrm{G}^{2} \rtimes \mathbb{F}_2$ for the hidden subgroup $K$. From our previous discussion it follows, if one shows that $K$ is indistinguishable from the identity subgroup, then QFS can not solve the required hidden subgroup problem. For an understanding of indistinguishability we refer to Kempe and Shalev~\cite{Kempe}.

\section{Niederreiter cryptosystems and Quasi-Cyclic Codes}
In this section we describe our Niederreiter cryptosystem which uses quasi-cyclic error correcting codes (QCC). Quasi-cyclic codes are a generalization of cyclic codes, codes where code-words are closed under right shifts. Since quasi-cyclic codes are a generalization of cyclic codes they can be expressed as nice algebraic objects. For more on QCC we refer Gulliver~\cite{Gulliver_thesis} and for algebraic structures of these codes to the work of Lally and Fitzpatrick~\cite{Lally}. An important underlying structure of a QCC is circulant matrices. Circulant  matrices are a building block of quasi-cyclic codes. 
\begin{definition} Circulant matrix: A $d \times d$ matrix $C^\prime$ over a field $F$ is called circulant if every row, except for the first row, is a circular right shift of the row above that.
\end{definition} 
A typical example of a circulant matrix is 
\begin{center}
$\left[ \begin{array}{cccc}
c_{0} & c_{1} & \cdots & c_{d-1} \\ 
c_{d-1} & c_{0} & \cdots & c_{d-2}\\
\vdots & \vdots & \ddots & \vdots\\
c_{1} & c_{2} & \cdots & c_{0} 
\end{array} \right]$. 
\end{center}
It is known that a circulant matrix over a field $F$ can be represented by its first row, as a polynomial of degree $d-1$ which is an element of the ring $F[x]/(x^d-1)$ and denoted by $f_{C^\prime}$ and called the representer polynomial. In this paper, we define the \emph{multiplicity} of a field element $a$ in $C^\prime$ as the number of times it appear in the first row of $C^\prime$. Multiplicity can be zero. For more on circulant matrices we refer to Davis~\cite{Davis}. An \emph{unique element} is an element in the first row of $C^\prime$ of multiplicity one.

\subsection{Conditions on the parity-check matrix}\label{cond_qcodes}
We need some terminology before we can describe our requirements on parity-check matrices of a quasi-cyclic codes. These condition are easy to generate and a large class of block circulant matrices satisfy these conditions.
\begin{definition} [Permutation equivalent rows and columns] Let $c_i,c_j$ be two column matrices, we say columns $c_{i}$ and $c_{j}$ are permutation equivalent if there is a permutation $\sigma \in\mathrm{S}_{n}$ such that $\sigma(c_{i})=c_{j}$. The permutation group acts on the indices of the columns. Similarly, if there is a permutation $\tau \in \mathrm{S}_n$ such that $\tau(r_{i})=r_j$, we say that rows $r_i$ and $r_j$ are permutation equivalent. 
\end{definition}
In short, two columns are permutation equivalent if one of them can can be reordered to get the other column. We now describe the quasi-cyclic code for our cryptosystem. We do this by stating conditions on the systematic parity-check matrix $\mathcal{H}$ for a quasi cyclic error-correcting code. The dimension of $\mathcal{H}$ is $m_1d\times m_2d$, where $m_1 < m_2$ are positive integers and $d\geq 2$. From now onward, we assume that $\mathcal{H}$ is of the form $[\mathrm{I}\,|\,C]$ where $\mathrm{I}$ is an identity matrix of size $m_1d$ and $C$ is a matrix with circulant blocks. Each block in $C$, is a circulant matrix $C_{ij}$ where $i=1,2,\ldots,m_1$ and $j=1,2,\ldots,(m_2-m_1)$ and is of size $d\geq 2$. It is well known that $C=\sum_{i,j} E_{i,j}\otimes C_{i,j}$ where $E_{ij}$ is the matrix with $1$ in the $(i, j)$ position and zero everywhere else.
\paragraph{Conditions on the parity-check matrix $\mathcal{H}$}
\begin{itemize}
\item[i)] There is one column in $C$, such that, all the circulants $C_{i,j}$ in that column have at least one unique element in the first row. Equivalently, there is a $j_0$, such that, for all $i$, circulant matrices $C_{i,j_0}$ will have at least one unique element in the first row. However, note that, different circulant blocks can share the same unique element.

\item[ii)] All matrices are defined over a field of characteristic 2. For each $j$ there is at least one $i$ where $C_{i,j}$ contains an element from a proper extension of $\mathbb{F}_2$. This condition is equivalent to saying that each column of $C$ contains at least one element from some non-trivial extension of $\mathbb{F}_2$. 

\item[iii)] Any two rows $r_i,r_j$ in $C$ are permutation equivalent only when $\left\lfloor \dfrac{i}{d} \right\rfloor =\left\lfloor \dfrac{j}{d}\right\rfloor$ where  $ 0 \leq i,j \leq m_1d$.  Similarly, two columns $c_i,c_j$ are permutation equivalent only when $\left\lfloor \dfrac{i}{d} \right\rfloor =\left\lfloor \dfrac{j}{d}\right\rfloor$ where $0 \leq i,j \leq m_2d$ .  

Note that if $\left\lfloor \dfrac{i}{d} \right\rfloor =\left\lfloor \dfrac{j}{d}\right\rfloor$ then rows $r_{i},r_j$ and columns $c_i,c_j$ are permutation equivalent because $C_{i,j}$ is a circulant matrix. So this condition simply says that apart from these permutation equivalences, there are no other permutation equivalence between columns or rows.

\item[iv)] There exist at least one $C_{i,j}$, a circulant matrix of size $d$, such that, the first row of the circulant matrix have two consecutive unique elements. Note that we consider the $d^\text{th}$ entry and the $1^\text{st}$ entry of the first row as consecutive. 
\end{itemize}   
\paragraph{Remark:}

Condition (iii) can be alternatively stated as there exists $a\in\mathbb{F}_{2^{\eta}}$, 
such that, the multiplicity of $a$ in $r_{i}$ is not the same as in $r_{j}$ where $\left\lfloor \dfrac{i}{d} \right\rfloor \neq \left\lfloor \dfrac{j}{d}\right\rfloor$. Equivalence of these two conditions follows directly as two rows are permutation equivalent if and only if one of them can be reordered to the other. Note that this equivalent condition of permutation equivalence is helpful in finding suitable $C$. This is because, if two rows are permutation equivalent, one can just count multiplicities rather that going through all possible permutations. Another way to construct such $C$ is to construct individual circulant matrices $C_{i,j}$ in such a way that each row of the individual circulant matrix has at least one entry different from every other circulant matrix.  

\subsection{From K to H} \label{sktoh}
Recall from Equation~\ref{hidden}, $\mathrm{G}=\mathrm{GL}_{k}(\mathbb{F}_{2})\times\mathrm{S}_n$
and we are trying to solve the hidden subgroup problem in the group $\mathrm{G}^2 \rtimes\mathbb{F}_{2}$. The subgroup in this case is $\mathrm{K}=K_0 \bigcup K_1$ 
where $\mathrm{K}_{0}=\left(\left( \mathrm{H}, s^{-1}\mathrm{H} s\right), 0\right)$ and $\mathrm{K}_{1}=\left(\left(\mathrm{H}s, s^{-1}\mathrm{H}\right), 1\right)$. Note that $\mathrm{H}$ replaces $H_0$ in Equation~\ref{hidden} and $s$ is the shift defined earlier, $K_1$ is not a subgroup and the union is disjoint. 

In this section, we reduce distinguishability of $\mathrm{K}$, i.e., $\mathcal{D}_{\mathrm{K}}$ to the subgroup $\mathrm{H}$. If we directly apply Equation~\ref{KS1.1}, our optimization should be over $\left(\mathrm{GL}_{k}(\mathbb{F}_{2})\times \mathrm{S}_n\right)^{2}\rtimes \mathbb{F}_{2}$. We reduce it to $\mathrm{H}$, a subgroup of $\mathrm{GL}_{k}(\mathbb{F}_{2})\times \mathrm{S}_n$. Then the later bound can be trivially decomposed into individual components $ \mathrm{GL}_{k}(\mathbb{F}_{2})$  and $\mathrm{S}_n$. Apart from getting rid of the $\mathbb{F}_{2}$ component it also serves one more, the most important purpose for further optimization. The bound that we develop are in terms of $\mathrm{H}$. It has no shift term $s$. It is well characterized by the subgroup $\mathrm{H}$ and has connections and structural properties that we exploit for our optimization. 
 
Back to some more notations.
Let $x^{\mathrm{G}}$ denote the conjugacy class of $x$ in $\mathrm{G}$ and $C_{\mathrm{G}}(x)$ denotes the centralizer of $x$ in $\mathrm{G}$. Here the group acts on itself by conjugation. A definition of these concepts are in Alperin and Bell~\cite[Pages 33-34]{bell}. These concepts can also be studied by group action, see Dixon and Mortimer~\cite[Example 1.3.5]{Dixon}.
 
From~\cite[Proposition 1 (2)]{KS2} we know that  
\begin{equation*} \mathcal{D}_{K} \leq \sum_{k\in \mathrm{K}, k\neq e} {\vert k^{\mathrm{G} ^2 \rtimes \mathbb{F}_2}\vert}^{-\frac{1}{2}} 
\leq \sum_{k_0\in \mathrm{K}_0, k_0\neq e} {\vert k_0^{\mathrm{G} ^2 \rtimes \mathbb{F}_2}\vert}^{-\frac{1}{2}} + \sum_{k_1\in \mathrm{K}_1} {\vert k_1^{\mathrm{G} ^2 \rtimes \mathbb{F}_2}\vert}^{-\frac{1}{2}}.
\end{equation*}
Where one $k$ is chosen from each orbit.
Let $S_0$ be the sum over $K_0$ and $S_1$ be the sum over $K_1$ in the above expression.

We present an upper bound for $\mathcal{D}_{K}$ by restricting our attention to $S_0$ and $S_1$. First we start with $S_1$. By the orbit-stabilizer property, $S_1$ can be rewritten as

\begin{equation}
 S_1 = \sum_{k_1\in \mathrm{K}_1}  \dfrac{\vert C_{\mathrm{G} ^2 \rtimes \mathbb{F}_2}(k_1)\vert ^\frac{1}{2}}{\vert \mathrm{G} ^2 \rtimes \mathbb{F}_2 \vert ^\frac{1}{2}}\leq \vert \mathrm{K}_1 \vert \ \dfrac{ \max\limits_{k_1\in\mathrm{K}_1}\vert C_{\mathrm{G} ^2 \rtimes \mathbb{F}_2}(k_1) \vert ^\frac{1} {2} }{ \vert \mathrm{G} ^2 \rtimes \mathbb{F}_2 \vert ^ \frac{1}{2}}. 
\end{equation}
Now we compute an upper bound of $C_{\mathrm{G} ^2 \rtimes \mathbb{F}_2}(k_1)$. Define two sets $\mathrm{G}_0$ and $\mathrm{G}_1$ as follows: 
 
 $$\mathrm{G}_{0} =\lbrace \left(\left(A,P\right)\left(A^\prime,P^\prime\right) ,0\right) \ :\ A,A^\prime\in \mathrm{GL}_{k}(\mathbb{F}_2); P,P^\prime \in \mathrm{S}_n\rbrace$$ and 
 $$\mathrm{G}_{1} =\lbrace \left(\left(A,P\right)\left(A^\prime,P^\prime\right),1\right)\ :\  A,A^\prime\in \mathrm{GL}_{k}(\mathbb{F}_2); P,P^\prime \in \mathrm{S}_n\rbrace.$$
Clearly $G^2\rtimes\mathbb{F}_2$ is the disjoint union of $\mathrm{G}_0$ and $\mathrm{G}_1$.
Then  $\vert C_{\mathrm{G} ^2 \rtimes \mathbb{F}_2}(k_1) \vert= \vert C_{\mathrm{G}_0}(k_1)\vert + \vert C_{\mathrm{G}_1}(k_1)\vert$; where $C_{\mathrm{G_1}}(k_1)=C_{\mathrm{G}}(k_1)\cap\mathrm{G}_1$ and $C_{\mathrm{G_0}}(k_1)=C_{\mathrm{G}}(k_1)\cap\mathrm{G}_0$. Define $g_0= \left(\left(A,P \right),\left(A^\prime,P^\prime\right) ,0\right)$. If $g_0$ is an element of $ C_{\mathrm{G}_0}(k_1)$ with $k_1= \left(\left(h_1s,s^{-1}h_2\right),1\right)$ it then follows that
$h_1s\left(A^\prime, P^\prime\right)=\left(A, P \right)h_1s$ where $h_1,h_2\in \mathrm{H}$. This is a simple calculation using the fact that $k_1$ commutes with $g_0$.

For each $\left(A ,P \right)$ and $h_1$, there is only one choice available for $\left(A^\prime,P^\prime\right)$. Thus 
\begin{equation}
     \vert C_{\mathrm{G}_0}(k_1)\vert \leq \vert \mathrm{H}\vert \vert \mathrm{GL}_k\left(\mathbb{F}_2\right) \times \mathrm{S}_n\vert.  \label{eqn:k1_x0} \end{equation}
Similarly, if we define $g_1=\left(\left(A,P\right),\left(A^\prime,P^\prime\right),1\right)$ and take into consideration that $k_1$ and $g_1$ commute we get that $h_1s\left(A^\prime,P^\prime\right)=\left( A ,P \right)s^{-1}h_2$. Similar arguments as above shows  
\begin{equation} 
 \vert C_{\mathrm{G}_1}(k_1)\vert \leq \vert \mathrm{H}\vert^{2}\vert \mathrm{GL}_k\left(\mathbb{F}_2\right) \times \mathrm{S}_n\vert.  \label{eqn:k1_x1} \end{equation}
Combining Equations~\ref{eqn:k1_x0} and~\ref{eqn:k1_x1}, we get 
$\vert C_{\mathrm{G} ^2 \rtimes \mathbb{F}_2}(k_1)\vert\leq  \left(\vert {H}\vert^{2}  + \vert \mathrm{H}\vert\right)\vert\mathrm{GL}_k\left(\mathbb{F}_2\right) \times \mathrm{S}_n\vert $.
Putting together above calculations along with $\vert \mathrm{K}_1 \vert = \vert \mathrm{H}\vert ^2$, we get
\begin{equation}
S_1\leq \vert\mathrm{H}\vert^{2}\left(\dfrac{\left(\vert \mathrm{H}\vert^{2}+\vert \mathrm{H}\vert\right)^\frac{1}{2}}{ \vert \mathrm{GL}_k\left(\mathbb{F}_2\right) \times \mathrm{S}_n\vert ^\frac{1}{2}}\right). \label{eq:S_1 upper  bound} \end{equation}  

We now look at $S_0$. From similar computation with $g_0 \in C_{\mathrm{G}_0}(k_0)$  and $k_0=\left(\left(h_1,s^{-1}h_2s\right), 0 \right)$ where $h_1,h_2\in\mathrm{H}$ we get the following two conditions:
\begin{itemize}
\item[i)] $h_1(A,P)=(A,P)h_1$ which implies that $\left( A,P\right) \in C_{\mathrm{G} }(h_1)$ 

\item[ii)] $s^{-1}h_2s(A^\prime,P^\prime)=\left(A^\prime,P^\prime \right)s^{-1}h_2s$ which implies that $s\left(A^\prime,P^\prime \right)s^{-1} \in C_{\mathrm{G}}(h_2)$.
\end{itemize}
Thus $C_{\mathrm{G}_0}(k_0)\leq\vert C_{\mathrm{G}}(h_1)\vert\vert C_{\mathrm{G}}(h_2)\vert$. Hence
\begin{equation}
    \dfrac{\vert C_{\mathrm{G}_0}(k_0)\vert}{\vert \mathrm{G}^2 \rtimes \mathbb{F}_2\vert} \leq \dfrac{\vert C_{\mathrm{G}}(h_1) \vert\vert C_{\mathrm{G}}(h_2)\vert}{\vert\mathrm{G}^2 \rtimes \mathbb{F}_2\vert} \leq \dfrac{\vert \mathrm{G} \vert \min( \vert C_{\mathrm{G}}(h_1) \vert, \vert C_{\mathrm{G}}(h_2) \vert )}{\vert \mathrm{G}^2 \rtimes \mathbb{F}_2\vert}. \label{eqn:k0x0}
\end{equation}
Now, let us assume that $g_1=\left((A,P),(A^\prime,P^\prime),1\right)\in C_{G_1}(k_0)$ where $k_0=\left(h_1,s^{-1}h_2s),0\right)$. Then $k_0g_1=\left(h_1(A,P),s^{-1}h_2
s(A^\prime,P^\prime),1\right)$ and $k_0^{-1}=(h_1^{-1},s^{-1}h_2^{-1}s,0)$. Now notice that $k_0g_1k_0^{-1}=\left(h_1(A,P)s^{-1}h_2^{-1}s,s^{-1}h_2s(A^\prime,P^\prime)h_1^{-1},1\right)$. Since $k_0$ and $g_1$ commute we get that
\begin{eqnarray*}
(A,P)&=&h_1(A,P)s^{-1}h_2^{-1}s\\
(A^\prime,P^\prime) & = & s^{-1}h_2s(A^\prime, P^\prime)h_1^{-1}
\end{eqnarray*} implying that $s^{-1}h_2s=(A^\prime,P^\prime)h_1(A^\prime,P^\prime)^{-1}$. Substituting in the above equation we get that $(A,P)=h_1(A,P)(A^\prime,P^\prime)h_1^{-1}(A^\prime,P^\prime)^{-1}$ or $(A,P)(A^\prime,P^\prime) = h_1(A,P)(A^\prime,P^\prime)h_1^{-1}$. This shows that 
\[(A,P)(A^\prime, P^\prime) \in C_{\mathrm{G}}(h_1).\]

Extracting $h_1$ from above and substituting into the other equation we see that $(A,P)=(A^\prime,P^\prime)^{-1}s^{-1}h_2s(A^\prime,P^\prime)(A,P)s^{-1}h_2^{-1}s$. This implies \[s(A^\prime,P^\prime)(A,P)s^{-1}\in C_\mathrm{G}(h_2).\]
Note that $h_1$ and $h_2$ are fixed. Then for each choice of $(A,P)$ there is at most $|C_{\mathrm{G}}(h_1)|$ choices of $(A^\prime, P^\prime)$. There are at most $|G|$ choices of $(A,P)$. Thus the total number of choices for $g_1$ is $|G||C_{\mathrm{G}}(h_1)|$.
These arguments remain valid when we switch $(A,P)$ with $(A^\prime,P^\prime)$ and $C_{\mathrm{G}}(h_1)$ with $C_{\mathrm{G}}(h_2)$. 

Thus $\vert C_{\mathrm{G}_1}(k_0) \vert\leq\min\left( C_{\mathrm{G}}(h_1),C_{\mathrm{G}}(h_2) \right)\vert \mathrm{G} \vert$ and hence
\begin{equation}
    \dfrac{\vert C_{\mathrm{G}_1}(k_0)\vert}{\vert \mathrm{G}^2 \rtimes \mathbb{F}_2\vert} \leq \dfrac{\vert \mathrm{G} \vert \min( \vert C_{\mathrm{G}}(h_1) \vert, \vert C_{\mathrm{G}}(h_2) \vert )}{\vert \mathrm{G}^2 \rtimes \mathbb{F}_2\vert}.
\label{eqn:k0x1}\end{equation}
Combining Equations~\ref{eqn:k0x0} and~\ref{eqn:k0x1}, we get 
\[ \dfrac{\vert C_{\mathrm{G}^2\rtimes \mathbb{F}_2}(k_0)\vert}{\vert \mathrm{G}^2 \rtimes \mathbb{F}_2\vert} = \dfrac{\vert C_{\mathrm{G}_0}(k_0)\vert}{\vert \mathrm{G}^2 \rtimes \mathbb{F}_2\vert}+\dfrac{\vert C_{\mathrm{G}_1}(k_0)\vert}{\vert \mathrm{G}^2 \rtimes \mathbb{F}_2\vert} \leq \dfrac{ \min( \vert C_{\mathrm{G}}(h_1) \vert, \vert C_{\mathrm{G}}(h_2) \vert )}{\vert \mathrm{G} \vert}.\]
Thus, 
\begin{eqnarray*}
S_0 &= &\sum_{k_0 \neq e} \left(\dfrac{\vert C_{\mathrm{G}^2\rtimes \mathbb{F}_2}(k_0)\vert}{\vert \mathrm{G}^2 \rtimes \mathbb{F}_2\vert}\right)^\frac{1}{2}\\
 &\leq &\sum_{(h_1,h_2)\neq (e,e)}\left( \dfrac{ \min( \vert C_{\mathrm{G}}(h_1) \vert, \vert C_{\mathrm{G}}(h_2) \vert )}{\vert \mathrm{G} \vert}\right)^\frac{1}{2}\\ 
&\leq & \sum_{h\in\mathrm{H}\setminus e} \vert \mathrm{H} \vert \left( \dfrac{\vert C_{G }(h)\vert}{\vert \mathrm{G} \vert}\right)^\frac{1}{2}.
\end{eqnarray*}
Again, from the orbit-stabilizer theorem,
 \begin{equation} 
 S_0 \leq \vert \mathrm{H} \vert \sum_{ h\in H\setminus e} {\vert h^{\mathrm{G}}\vert}^{-\frac{1}{2}}
     \label{eq:S_0 upper  bound}
 \end{equation}
and thus we have achieved our goal for this section of writing $\mathcal{D}_\mathrm{K}$ in terms $\mathrm{H}$. In particular, this can be done by putting multiplicative overhead for $\vert \mathrm{H}\vert$ and an additive term given by Equation~\ref{eq:S_1 upper  bound}. Thus
\begin{equation} \mathcal{D}_{\mathrm{K} } \leq \vert \mathrm{H} \vert \sum_{ h\in H\setminus e} {\vert h^{\mathrm{G}}\vert}^{-\frac{1}{2}} +  \vert\mathrm{H}\vert^{2}\left(\dfrac{\left(\vert \mathrm{H}\vert^{2}+\vert \mathrm{H}\vert\right)^\frac{1}{2}}{ \vert \mathrm{GL}_k\left(\mathbb{F}_2\right) \times \mathrm{S}_n\vert ^\frac{1}{2}}\right).
    \label{k to h}
\end{equation}
\section{Size and minimal degree of $\mathrm{H}$} \label{sec:size}
Note that in the previous section, we have boiled down the indistinguishability of $\mathrm{K}$ to conjugacy classes of $\mathrm{H}$. Similar to the work of Kempe and Shalev~\cite{Kempe}, minimal degree and size of the subgroup play an important role in showing indistinguishability of $\mathrm{K}$. In this section, we give an upper bound on the size of $\mathrm{H}$ and determine their minimal degrees. 
Before that, we recall some well known definitions.

\begin{definition} For any group $\mathrm{G}\leq \mathrm{G}_{1} \times \mathrm{G}_{2}$ we define $\Pi^{i}(\mathrm{G})$ as a projection of the group $\mathrm{G}$ on $\mathrm{G}_i$ for $i=1,2$. 
\end{definition}
\begin{definition} Let $M_{k,n}$ be the set of $k \times n$ matrix. Then there is a natural group action of $\mathrm{S}_k \times \mathrm{S}_n$ on $M_{k,n}$ given by $(P_1,P_2)M=P_1^{-1}MP_2$. Let $\mathrm{Stab}(C)$ be the stabilizer of $C$. Furthermore, $T_C :=\Pi^{1}(\mathrm{Stab}(C))$.
 \end{definition}
\noindent The main theorem for this section is the following:
 \begin{theorem} Let $\mathcal{H}$ be a parity check matrix for a code that satisfies conditions in Section~\ref{cond_qcodes}. Recall that we define $H=\left\{(A,P)\in\mathrm{GL}_k(\mathbb{F}_2)\times\mathrm{S}_n\,:\,A^{-1}\mathcal{H}P=\mathcal{H}\right\}$. Then the following is true.
 \begin{itemize}
 \item[(i)] $\vert \mathrm{H}\vert \leq d$
 \item[(ii)] The minimal degree of $\Pi^{1}(\mathrm{H}) \geq d$ and the minimal degree of $\Pi^{2}(\mathrm{H}) \geq d$.
 \end{itemize}
 \end{theorem}
\noindent To prove the above theorem, we need a key lemma. We prove that next.
\begin{lemmaA} Let $\mathcal{H}$ be a parity-check matrix such that it satisfies conditions from Section~\ref{cond_qcodes} then
$T_C \hookrightarrow S_d \times S_{d} \times \cdots \times S_d \times \langle(1,2,\ldots,d)\rangle \times S_d \times \cdots\times S_d$. The direct product is taken over $m_1-1$ terms of $S_d$ and one term of the cyclic group of size $d$.
\end{lemmaA}
This establishes upper and lower bounds on the size and the minimal degree of $T_C$. Later, we will translate this to that of $\mathrm{H}$. 
\begin{lemma} \label{S_k in}
 Let $(A,P) \in  \mathrm{H}$ then 
 \begin{itemize}
\item[] $A = P_1 $  
\item[] $P = A^{-1} \oplus P_2 = P_1^{-1} \oplus P_2$.
\end{itemize} 
 where $P_1 \in \mathrm{S}_k$ and $P_2 \in S_{n-k}$. Moreover, $P_1CP_2=C$ and for each $P_1$ there is an unique $P_2$. It then follows, $T_C=\Pi^{1}(\mathrm{H})$ and $\vert T_C\vert =\vert \mathrm{H}\vert$. 
 \end{lemma}
\begin{proof} 
Let $(A,P) \in \mathrm{H}$ then by definition we have \[[\mathrm{I}|C]= A [\mathrm{I}|C] P= [A|AC]P. \]

Since action of right multiplication by $P$ is equivalent to reordering of columns we infer that $[A|AC]$ and $[\mathrm{I}|C]$ have same columns possibly reordered. By construction (in particular condition (ii) in Section~\ref{cond_qcodes}), $C$ and the identity matrix $\mathrm{I}$ have no common columns as every column of $C$ contains an element from proper extension. As $C$ and $\mathrm{I}$ have distinct columns; $A$ should have same columns as the identity matrix $\mathrm{I}$. So by the action of multiplication by $P$ first $k$ columns must go to themselves, in other words, first $k$ columns make up a permutation matrix of size $k$.  Hence $P$ is a block diagonal matrix, having a block of size $k$ and $n-k$ where each of the blocks is a permutation matrix of size $k$ and $n-k$ respectively, we get $P= \sigma_k \oplus \sigma_{n-k}$. Now $A\sigma_k =\mathrm{I}$  gives  $A=P_1$ and $P = A^{-1} \oplus P_2$ where $P_1 = {\sigma_k}^{-1}$ and $P_2=\sigma_{n-k}$. It is easy to see from the fact $\mathcal{H}=[I\,|\,C]$ that $P_1CP_2=C$. 

Clearly, $T_C= \Pi^{1}(\mathrm{H})$ as $T_C$ being a subgroup, it is closed under inverse. Now, to show uniqueness of $P_2$ for every $P_1$, it suffices to prove for every $P_1$ there is at most one $P_2$. This follows from $P_1CP_2=C$ because no two columns of $C$ are identical and so no two columns of $P_1 C$ are identical. Now $P_2$ should reorder the columns to give back $C$ which can be done at most in one way. Hence, for every $P_1$ there is at most one corresponding $P_2$. 
\end{proof}    
Now we move to find an upper bound for the size of $\mathrm{H}$ by embedding it into direct product of $m_1$ full symmetric groups and the cyclic group $\langle(1,2,\ldots,d)\rangle$.
\begin{lemma} Let $P,Q$ be permutation matrices such that $PCQ=C$, then $P= \sum_{i}E_{i,i} \otimes P_{i}$ and $Q=\sum_{j}E_{j,j}\otimes Q_{j}$ where $P_{i}$ and $Q_{j}$ are permutation matrices of size $d$.
\end{lemma}
\begin{proof}
Note that the lemma simply says that all $P$ and $Q$ are block diagonal permutation matrices with blocks of size $d$. We prove the decomposition of $Q$, a similar result for $P$ can be achieved by similar arguments.

Suppose there exists a $Q \in \Pi^{2}(\mathrm{Stab}(C))$ that can not be decomposed into the block diagonal form. Then there is some $i$,$j$ such that $Q(i)=j$ and $\left\lfloor \dfrac{j}{d}  \right\rfloor \neq \left\lfloor  \dfrac{i}{d} \right\rfloor$ (corresponding to off block diagonal entry at $P_2(i,j)$). Now by condition (iii) on $C$, $c_i$ and $c_j$ are not permutation equivalent. And thus $i^{th}$ column of $CQ$ and $i^{th}$ column of $C$ are not permutation equivalent. Thus for any $P$, the $i^{th}$ column of $PCQ$ can not be equal to $i^{th}$ column of $C$. Thus for any permutation $P$, $PCQ \neq C$. Which leads to a contradiction.   
\end{proof}

\begin{lemma} {\label{Tc break}} The group  $T_C \hookrightarrow T_{C_{1r_1}} \times T_{C_{2r_2}} \times T_{C_{3r_3}}\times \cdots \times T_{C_{m_1r_{m_1}}}$ for all $r_i \in \lbrace{1,2,\ldots,m_2-m_1 \rbrace}$. 
\end{lemma}
\begin{proof} By the decomposition above, it follows that for every $i,j$, we have $P_iC_{i,j}Q_j= C_{i,j}$.
\begin{eqnarray*}
PCQ & = (\Sigma_{k} E_{kk} \otimes P_{k}) (\Sigma_{i,j} E_{i,j} \otimes C_{i,j}) (\Sigma_{l} E_{l,l} \otimes Q_{l}) \\
& = \Sigma_{i,j,k,l} (E_{kk}E_{i,j}E_{ll} \otimes P_{k}C_{i,j}Q_{l})\\
& = \Sigma_{i,j}(E_{i,j} \otimes P_{i}C_{i,j}Q_{j}). \end{eqnarray*}
The canonical map sending $P \to \left(P_{1}, P_{2}, P_{3}, \ldots, P_{m_{1}}  \right)$ gives the required inclusion. 
\end{proof}
\begin{remark} \label{ok now} Note that until now, we have used conditions ii) and iii). So, for any $C$ satisfying those conditions, Lemma~\ref{Tc break} is valid.  \end{remark}

We denote the particular $C_{i,j}$ satisfying condition (iv) by  $C^\prime$ to increase readability.  We show that if the first row of $C^\prime$ has two consecutive unique elements then $T_{C^\prime}$ is the cyclic group generated by the permutation $(1,2,\ldots,d)$. 
Recall that 
\[T_{C^\prime}=\{P\in S_d \vert\,Q\in S_d\,\text{and}\, PC^\prime Q=C^\prime\}.\]
\begin{theorem}
If $C^\prime$ has two consecutive unique elements in the first row then $T_{C^\prime}=\langle(1,2,\ldots,d)\rangle$.

\end{theorem}
\begin{proof}
The proof follows from the fact that in a circulant matrix with two consecutive unique elements, an unique triangle is formed when two consecutive rows are taken into account. For sake of exposition let us call these two unique elements $a$ and $b$ where $a$ occurs left of $b$ except in the isolated case where $a$ occurs as the last element and $b$ the first element. If we look at the first two rows, there is a unique position for $a,b$ and then the $a$ occurs below $b$ in the circulant matrix $C^\prime$. We will only concentrate on this triangular pattern in this proof.

Now we concentrate on the first row and assume that $P(1)=k_1$. Now $P$ maps the unique elements $a,b$ to the row $k_1$. Then $Q$ maps these columns that contain this $a,b$ to columns that preserves the circulant nature of $C^\prime$. Now assume that $P(2)=k_2$. Now we look at $a$ an element in the second row. Since it occurs right below $b$, $Q$ will move it to the same column as it moved $b$. Furthermore, the action of $Q$ will preserve the circulant nature of $C^\prime$. Now notice that the unique triangle must be preserved. This implies that $k_2=k_1+1$. This completes the proof.
\end{proof}
\begin{remark}
It is easy to see that if there is an unique element and another element repeating in the rest of the places then $T_{C^\prime}$ is the full symmetric group. In the above theorem, two consecutive unique elements is sufficient for $T_{C^\prime}=\langle(1,2,\ldots,d)\rangle$. One question arises, is this condition necessary? This translates to, are there any other pattern that give rise to $T_{C^\prime}=\langle (1,2,\ldots,d)\rangle$? Furthermore, we have noticed while doing computer experiments that in the case of $d=5$, for some repeat patterns in the first row of $C^\prime$, $T_{C^\prime}$ is the dihedral group of size $10$. For $d=7$ it was a much larger group. So the obvious question arises, can one classify $T_{C^\prime}$ based on the repeat pattern in the first row of $C^\prime$. 
\end{remark}
So, the required decomposition $T_C \hookrightarrow S_d \times S_{d} \times \cdots \times S_d \times \langle(1,2,\ldots,d)\rangle \times S_d \times \cdots \times S_d$ follows. This proves the subgroup decomposition lemma. Now we are in a position to reach the main theorem using condition (i).

\begin{lemma} \label{for_all_exactly_1} Let $C_{ij}$ be matrices as in condition (i). For all $i,j$ we have, for every $P$ there is at most one solution $Q$ such that $P C_{i,j} Q=C_{i,j}$.
Similarly, for each $Q$ there is at most one $P$ such that $P C_{i,j}Q=C_{i,j}$. 
\end{lemma}
\begin{proof}
Since no two columns of $C_{i,j}$ are identical after the action of $P$, no two columns of $PC_{i,j}$ are identical and then there is unique $Q$ that can reorder columns to get back $C_{i,j}$. Similar row argument proves the unique $P$.
  \end{proof}
\begin{lemma}\label{one-to-one}Let $P_1, P_2 \in T_c$ where $P_1= \left( P_{11},P_{12},\ldots P_{1 m_1}\right)$ and $P_2 = \left( P_{21},P_{22},\ldots, P_{2m_{1}}\right)$. If $P_{1i} = P_{2i}$ for some $i$, then $P_1 = P_2$. 
\end{lemma}
\begin{proof}
Suppose there exists $i_0$ such that $P_{1i_{0}}=P_{2i_{0}}$.
By condition (i), there is a $j_0$, such that,  $P_{1i_{0}}C_{i_{0}j_0}Q_{1j_0}=P_{2i_{0}} C_{i_{0}j_0}Q_{2j_0}=C_{i_{0}j_0}$ for some permutation matrices $Q_{1j_0},Q_{2j_0}$. 
From above, we get $Q_{1j_0}=Q_{2j_0}$.
 Again apply the same lemma on $C_{ij_0}$ for any $i$ and we get that $P_{1i}=P_{2i}$ for all $i$. 
Thus, we have proved that if for some $i_{0}$, $P_{1i_{0}}= P_{2i_{0}}$ then for all $i$, $P_{1i}=P_{2i}$ and hence $P_1=P_2$.
\end{proof}

Note, this means that there is an injective mapping from the group $T_{C}$ to the group $T_{C_{i r_{i}}}$ for any component in the above decomposition. In particular, if we choose $C_{ir_i}$ as a matrix satisfying condition (iv) then that corresponding component is the cyclic group $\langle(1,2,\ldots,d)\rangle$. Thus we get that $\vert T_C \vert \leq \vert T_{C^\prime} \vert \leq d$.

\begin{corollary} \label{min} The minimal degree of $\Pi^{1}(\mathrm{H})$ and $\Pi^{2}(\mathrm{H})$ is at least $d$.
\end{corollary} 
\begin{proof}  Clearly, from $T_C=\Pi^{1}(\mathrm{H})$ in Lemma~\ref{S_k in}, we get the minimal degree of $\Pi^{1}(\mathrm{H})$ to be at least $d$. Because, non-identity elements of $T_C$ must be non-identity in the cyclic component of the direct product decomposition (by Lemma \ref{one-to-one}). Moreover, non-identity elements of the cyclic component can not fix any element. So, the minimal degree of $T_C=\Pi^{1}(\mathrm{H})$ is at least $d$. For minimal degree of $\Pi^{2}(\mathrm{H})$ we show that the minimal degree of $\Pi^{2}(\mathrm{H})$ is at least as large as that of $\Pi^{1}(\mathrm{H})$. Note that, if $P_{1} \oplus P_{2} \in \Pi^{2}(\mathrm{H})$ is a non-identity element then $P_{1} \neq \mathrm{I}$. This follows from the uniqueness of $P_2$ in Lemma~\ref{S_k in}. Now, the non-identity element $P_1$ has support of at least the size of the minimal degree of $\Pi^{1}(\mathrm{H})$ as it is also an element of $\Pi^{1}(\mathrm{H})$. 

Thus minimal degree of $\Pi^{2}(\mathrm{H})$ is greater than or equal to the minimal degree of $\Pi^{1}(\mathrm{H}) \geq d$. 
\end{proof}
\section{Subgroup $K$ is indistinguishable}
In this section, we collect results from the last section on bounds on projections of $\mathrm{H}$ into Equation~\ref{hidden} to show that $\mathrm{K}$ is indistinguishable.  
Recall that $\mathrm{G} = \mathrm{GL}_{k}\left( \mathbb{F}_{2}\right) \times \mathrm{S}_n$ and $\mathrm{H}=\lbrace{ (A,P) : A^{-1}\mathcal{H}P = \mathcal{H} \rbrace} $. Now from the discussion in previous section, we know that $A$ is a permutation matrix.

Let $h =(\sigma_1,\sigma_2) \in \mathrm{H}$. Then

\[ {\vert h^{\mathrm{G}}\vert} ^{-\frac{ 1}{2}} = \dfrac{ \vert C_{\mathrm{G}}(h) \vert^\frac{1}{2}}{\vert\mathrm{G}\vert ^ \frac{1}{2}} = \left(\dfrac{C_{\mathrm{GL}_{k}(\mathbb{F}_2)}(\sigma_1)}{ \vert \mathrm{GL}_{k}(\mathbb{F}_2) \vert} \right)^ \frac{1}{2} \left( \dfrac{C_{\mathrm{S}_n}(\sigma_2)}{\vert \mathrm{S}_n \vert}\right)^\frac{1}{2} \]

\[ = {\vert \sigma_1^{\mathrm{GL}_k(\mathbb{F}_2)}\vert}^{-\frac{1}{2}} {\vert \sigma_2^{\mathrm{S}_n} \vert} ^{-\frac{1}{2}} \leq {\vert \sigma_1^{\mathrm{S}_k}\vert }^{-\frac{1}{2}} {\vert \sigma_2^{\mathrm{S}_n} \vert} ^{-\frac{1}{2}}. \]
The last inequality follows because the conjugacy class of $h$ in a permutation group is a subset of a conjugacy class in the general linear group. Also note, if $h$ is not the identity in $\mathrm{H}$ then
$\sigma_1 \neq \mathrm{I}$ and $\sigma_2 \neq \mathrm{I}$ by the uniqueness property in Lemma~\ref{S_k in}.

From Equation~\ref{k to h} we have, 	
\begin{equation} \label{hlimit}
\sum_{h\neq e} {\vert h^{\mathrm{G}}\vert} ^{-\frac{ 1}{2}} \leq \sum_{\sigma_1,\sigma_2\neq e}{\vert \sigma_1^{\mathrm{S}_k}\vert }^{-\frac{1}{2}} {\vert \sigma_2^{\mathrm{S}_n} \vert} ^{-\frac{1}{2}} =  \sum_{\sigma_1\in \Pi^{1}(\mathrm{H})\setminus e} {\vert \sigma_1^{\mathrm{S}_k}\vert }^{-\frac{1}{2}} \sum_{\sigma_2 \in \Pi^{2}(\mathrm{H})\setminus e}  {\vert \sigma_2^{\mathrm{S}_n} \vert} ^{-\frac{1}{2}}. \end{equation} 
We present this for sum over $\sigma_1$. A similar result can be obtained for $\sigma_2$.

Let $\Gamma_t$ denote the set of elements of $\mathrm{S}_k$ of support $t$. Then from a well-known theorem~\cite[Theorem B]{KS2} it follows that there exists absolute constants $b, \varepsilon$ such that if $\Pi^{1}(\mathrm{H})$ has minimal degree greater than $\delta  \geq b$ then  
$\vert \Gamma_t \vert \leq k^\frac{-\varepsilon \delta}{2} {{k}\choose {t}} ^\frac{1}{2} \left( t! \right)^\frac{1}{4}$.  

From another well known theorem~\cite[Lemma 8]{KS2}, we know that if $\mathcal{C}$ is a conjugacy class of elements of support $t$ inside $\mathrm{S}_k$. Then $\vert \mathcal{C} \vert \geq c {{k}\choose{t}} \sqrt{t!} t^{-\frac{1}{2}}$ where $c$ is some positive absolute constant. 
Therefore,
\[  \sum_{\sigma_1\in \Gamma_t} {\vert {\sigma_1^{\mathrm{S}_k}}\vert}^{-\frac{1}{2}} \leq c^{-\frac{1}{2}} \vert \Gamma_t\vert  {k \choose t} \left(k!\right)^{-\frac{1}{4}} k^{\frac{1}{4}}.\]
This gives,
\[ \sum_{\sigma_1 \in \Pi^{1}(\mathrm{H})\setminus e }\vert {\sigma_1^{\mathrm{S}_k} \vert }^{-\frac{1}{2}} = \sum_{t=\delta}^{k} \sum_{{\sigma_1}\in {\Gamma_t}} {\vert {\sigma_1^{\mathrm{S}_k}}\vert}^{-\frac{1}{2}} \]
\[ \leq \sum_{t=\delta}^{k}c^{-\frac{1}{2}} \vert \Gamma_t\vert  {k \choose t} \left(k!\right)^{-\frac{1}{4}} k^{\frac{1}{4}}. \]
Substituting,
 \[ \sum_{\sigma_1 \in \Pi^{1}(\mathrm{H})\setminus e }\vert {\sigma_1^{\mathrm{S}_k} \vert }^{-\frac{1}{2}} \leq \sum_{t=\delta}^{k}  c^{-\frac{1}{2}} k^{-\varepsilon \delta} k^{\frac{1}{4}} \leq a_k k^{ -\varepsilon \delta} k^{\frac{5}{4}} \] for some constant $a_k \geq 0$.
Similarly we can get  an upper bound for the other sum. Thus, we have 
 \begin{eqnarray*}
 \sum_{\sigma_1 \in \Pi^{1}(\mathrm{H})\setminus e }\vert {\sigma_1^{\mathrm{S}_k} \vert }^{-\frac{1}{2}}& \leq &  a_k k^{ -\varepsilon \delta_{1}} k^{\frac{5}{4}}\\
 \sum_{\sigma_1 \in \Pi^{1}(\mathrm{H})\setminus e }\vert {\sigma_1^{\mathrm{S}_n} \vert }^{-\frac{1}{2}} &\leq &  a_n n^{ -\varepsilon \delta_{2}} n^{\frac{5}{4}}
 \end{eqnarray*}
where $\delta_1,\delta_2$ are minimal degrees of $\Pi^{1}(\mathrm{H})$ and $\Pi^{2}(\mathrm{H})$. Putting this in Equation~\ref{hlimit}, we get
\begin{equation}
    \label{pp}
 \sum_{h\neq e} {\vert h^{\mathrm{G}}\vert} ^{-\frac{ 1}{2}} \leq  a_k a_n k^{ -\varepsilon \delta_{1}} k^{\frac{5}{4}} n^{ -\varepsilon \delta_2} n^{\frac{5}{4}}.\end{equation} 
 \begin{proof}[\textbf{Proof of Theorem A}]\leavevmode\\ 
 To prove $K$ is indistinguishable, we need to show that
$\mathcal{D}_{K} \leq \left(\log(\vert \mathrm{G}^{2} \rtimes \mathbb{F}_2 \vert)\right)^{-c}$  for every positive constant $c$.
 From Equation~\ref{k to h}, it suffices to prove that,
\[\log\left( \vert \mathrm{H} \vert \sum_{h\neq e} {\vert h^{\mathrm{G}}\vert} ^{-\frac{ 1}{2}}+   \frac{ \vert \mathrm{H}\vert \left(\vert \mathrm{H}\vert + \vert \mathrm{ H }  \vert^2\right)^{\frac{1}{2}}}{{\vert \mathrm{G}\vert}^{\frac{1}{2}}} \right)\leq\log(\Delta_c)\]
where $\Delta_c= \left(\log(\vert \mathrm{G}^{2} \rtimes \mathbb{F}_2 \vert)\right)^{-c}$.

Now
\begin{eqnarray*}\label{last}
& &\log\left( \vert \mathrm{H} \vert \sum_{h\neq e} {\vert h^{\mathrm{G}}\vert} ^{-\frac{1}{2}}+\frac{ \vert \mathrm{H}\vert \left(\vert \mathrm{H}\vert + \vert \mathrm{H}  \vert^2\right)^{\frac{1}{2}}}{{\vert \mathrm{G}\vert}^{\frac{1}{2}}} \right)\\ 
&\leq & \log\left( 2 \max \left\lbrace \vert \mathrm{H} \vert \sum_{h\neq e} {\vert h^{\mathrm{G}}\vert} ^{-\frac{ 1}{2}},   \frac{ \vert \mathrm{H}\vert \left(\vert \mathrm{H}\vert + \vert \mathrm{ H }  \vert^2\right)^{\frac{1}{2}}}{{\vert \mathrm{G}\vert}^{\frac{1}{2}}} \right\rbrace  \right)\\
&=& \log(2) + \log\left(  \max \left\lbrace \vert \mathrm{H} \vert \sum_{h\neq e} {\vert h^{\mathrm{G}}\vert} ^{-\frac{ 1}{2}},   \frac{ \vert \mathrm{H}\vert \left(\vert \mathrm{H}\vert + \vert \mathrm{ H }  \vert^2\right)^{\frac{1}{2}}}{{\vert \mathrm{G}\vert}^{\frac{1}{2}}} \right\rbrace  \right).
\end{eqnarray*}
Putting $\vert \mathrm{H}\vert \leq d$ and $\delta_1,\delta_2 \geq d$  and Equation~\ref{pp} one can verify that the above term is less than $\log(\Delta_c)$ for large enough $d$.

This completes our proof of indistinguishability of the subgroup $K$, making the cryptosystem resistant to hidden subgroup attacks.
\end{proof}

\section{Conclusion}
Niederreiter cryptosystems using quasi-cyclic codes are popular these days. The main reason behind this interest is quantum-security. This makes it a good candidate for post-quantum cryptography. This is evident from the NIST submissions~\cite{nist1,nist2,BIKE}.

Historically speaking, post-quantum cryptography grew out of Shor's algorithm to factor integers which was later used to solve the discrete logarithm problem. These algorithms use the hidden subgroup problem in finite abelian groups. This hidden subgroup problem follows from the scrambler-permutation problem in the non-commutative setting which uses complex irreducible representations of the group. If this hidden subgroup is indistinguishable by the quantum Fourier sampling from the identity subgroup then we can not solve the corresponding scrambler-permutation problem. This makes Niederreiter cryptosystem quantum secure. The idea behind the hidden subgroup problem for Niederreiter cryptosystem was put forward by Dinh et al.~\cite{Dinh} and the idea of distinguishability of subgroups was put forward by Kempe and Shalev~\cite{Kempe}.

We prove that for a Niederreiter cryptosystem using quasi-cyclic codes, satisfying certain conditions, the corresponding hidden subgroup is indistinguishable from the identity subgroup by weak quantum Fourier sampling. 

\section*{Acknowledgements} The first author was supported by the Singapore Ministry of Education and the National Research Foundation through the core grants of the Centre for Quantum Technologies. The second author was supported by a MATRICS research grant of SERB, govt.~of India and by a NBHM research grant. Both the authors thank Subhamoy Maitra and Joachim Rosenthal for insightful comments. 

We the authors owe a debt of gratitude to both the referees for their thorough reading of the manuscript and insightful comments. This has improved the content of this paper and its presentation.

\begin{small}
\bibliography{ref.bib}
\end{small}
\Addresses
\end{document}